\documentclass[11pt]{article}
\usepackage{amsmath,amsfonts,amssymb,amsthm}
\usepackage{graphicx}
\usepackage{tikz}
\usepackage{enumerate}
\usepackage[hyphens]{url}
\usepackage{xcolor}
\usepackage{url}
\usepackage{caption}
\usepackage{subcaption}
\usepackage{tcolorbox}
\usepackage{hyperref}
\usepackage{chemarr}
\usepackage[margin=3.5cm]{geometry}
\usepackage{capt-of}
\usepackage{authblk}
\usepackage{booktabs}

\usepackage{algorithm}
\usepackage{algpseudocode}
\usepackage{float}
\usepackage[numbers]{natbib}
\usepackage{svg}

\usepackage{listings}
\theoremstyle{plain}
\newtheorem{theorem}{Theorem}
\newtheorem{proposition}[theorem]{Proposition}

\theoremstyle{definition}
\newtheorem{definition}[theorem]{Definition}
\newtheorem{example}[theorem]{Example}
\newtheorem{remark}[theorem]{Remark}

\definecolor{darkgreen}{rgb}{0.0,0.7,0.0}

\usepackage{tikz}
\usetikzlibrary{patterns}
\usetikzlibrary{matrix, arrows, decorations.pathmorphing}
\usepackage{tikz-cd}
\tikzset{commutative diagrams/.cd}
\usepackage{framed,lipsum}
\usepackage{float}
\usepackage{pgfplots}
\usetikzlibrary{patterns}
\usetikzlibrary{intersections}
\usetikzlibrary{positioning, shadows, arrows}
\tikzstyle{every node}=[anchor=west, minimum height=3em]
\usetikzlibrary{decorations.pathmorphing}

\newcommand{\mS}{\mathcal{S}}
\newcommand{\mC}{\mathcal{C}}
\newcommand{\mR}{\mathcal{R}}
\newcommand{\SCR}{(\mathcal{S}, \mathcal{C}, \mathcal{R})}
\newcommand{\SCRk}{(\mathcal{S}, \mathcal{C}, \mathcal{R}, \boldsymbol{k})}
\newcommand{\bk}{\boldsymbol{k}}
\newcommand{\bx}{\boldsymbol{x}}

\begin{document}

\title{Implementation of Linear Regression and Linear Interpolation using Reaction Networks}

\author[1]{Aryan Kumar}
\author[2]{Amey Choudhary}
\author[3]{Jiaxin Jin}
\author[4]{Chittaranjan Hens}
\author[5]{Abhishek Deshpande}
\affil[1,2,4,5]{\small Center for Computational Natural Sciences and Bioinformatics, International Institute of Information Technology, Hyderabad} 
\affil[3]{\small Department of Mathematics, University of Louisiana at Lafayette}

\maketitle

\begin{abstract}

Statistical inference is a fundamental component of data science. In this work, we focus on two classical inference techniques: regression and interpolation. We propose a reaction-network-based framework for implementing linear regression, including both univariate and multivariate settings, as well as linear interpolation. Our approach encodes the outputs of these inference techniques in the steady-state concentrations of species within the reaction network. A key ingredient of the construction is a novel generalized division module capable of handling division involving negative numbers. We validate the proposed framework through in silico implementations on standard synthetic datasets and obtain the expected regression and interpolation outputs.

\vspace*{10pt}
\vspace*{10pt}

\end{abstract}

\section{Introduction} \label{sec1}

Nature possesses an inherent capacity to perform computation. Living cells continuously process information, respond to environmental stimuli, and maintain their internal biological states \cite{Kitano2004vw}. For example, proteins fold into their optimal shape within milliseconds \cite{KUBELKA200476}, a task that remains computationally expensive for modern algorithms \cite{levinthals_paradox}. These observations suggest that natural systems are capable of performing computation. Since Adleman’s solution of the NP-complete directed Hamiltonian path problem using DNA~\cite{adleman1994}, programmable biological platforms such as DNA strand displacement cascades~\cite{Simmel2019} have emerged as powerful frameworks for implementing computations based on chemical reaction networks (CRNs).

CRNs perform computation by representing values as molecular concentrations and arithmetic operations through specific reaction pathways. Such molecular circuits are particularly attractive in settings where traditional silicon-based computers are impractical or unavailable. 
Furthermore, molecular systems naturally support massive parallelism~\cite{Gines2023}, while DNA-based platforms can store information using orders of magnitude less physical space than conventional storage technologies~\cite{ceze2019molecular}. Processing information directly within a DNA-based system eliminates the need to convert molecular signals into digital representations, thereby reducing intermediate processing steps.

Recently, there has been significant interest in implementing machine learning algorithms \cite{gopalkrishnan2016scheme, pei2010training} using chemical reaction networks. Prior work has explored the implementation of neural networks \cite{anderson2021reactionnetworkimplementationsneural, fan2023automatic, gunawardena2003chemical} and various approaches to reservoir computing \cite{10.1145/2554797.2554827, 10937022_resv}. 
In contrast, statistical inference tasks such as linear regression and interpolation require the manipulation of signed real-valued parameters. Although dual-rail encoding provides a standard representation of signed quantities, there is limited literature describing explicit reaction-network implementations of such algorithms. In particular, existing CRN division modules are restricted to positive concentrations, which creates a fundamental obstacle to general real-valued computation and parameter estimation.

To address these limitations, we present a comprehensive, dual-rail framework for performing univariate linear regression, multivariate linear regression, and linear interpolation, entirely within the dynamics of CRNs. Moving beyond regression as a simple perceptron abstraction, we provide explicit reaction-network constructions that realize these algorithms through synchronized reaction pathways.
A central contribution of this work is a generalized division module, which is described in Section~\ref{sec:generalized-div}. We extend division to arbitrary-sign real values within the dual-rail encoding by decoupling magnitude computation from sign determination via an approximate majority network. This enables general real-valued arithmetic in CRNs.

Furthermore, to prevent interference between intermediate species during sequential operations under mass-action kinetics, we employ a Hopf oscillator (see Section \ref{sec:oscillation_module}) to enforce discrete phase-clocking. To improve efficiency and reduce transient oscillations, we introduce structural optimizations, including the combination of multiplication and addition modules (see Section \ref{sec:combined-mul-add}). Finally, we establish the theoretical correctness of the proposed networks and validate their performance through simulations on synthetic datasets. We believe that the proposed framework lays a robust structural foundation for future physical implementations.

\paragraph{Structure of the paper:}

In Section~\ref{sec:reaction_networks}, we introduce reaction networks and the mass-action systems. In Section~\ref{sec:modules}, we present reaction-network implementations of fundamental modules. We also employ dual-rail encoding to represent both positive and negative values arising in the computations. 
In Section~\ref{sec:div-module}, we present a generalized division module for arbitrary-sign real numbers.
Using these modules, we develop algorithms for univariate and multivariate linear regression via direct formulation and gradient descent, respectively, in Section~\ref{sec:linear_regression}, and linear interpolation in Section~\ref{sec:linear_interpolation}.
We validate all proposed reaction-network methods through simulations on synthetic and standard datasets.
Finally, in Section~\ref{sec:discussion}, we summarize the framework and discuss future directions.

\section{Chemical Reaction Networks}
\label{sec:reaction_networks}

In this section, we recall some basic notions from reaction network theory. Towards this, we first define a chemical reaction network.

\begin{definition}[\cite{feinberg}]

A \textit{chemical reaction network} is a triple $\SCR$ where
\begin{itemize}
    \item $\mS = \{X_1, \dots, X_n \}$: a set of species
    \item $\mC = \{ C_1, \dots, C_s \}$: a set of complexes.
    \item $\mR = \{ R_1, \dots, R_r \}$: a set of reactions.
\end{itemize}

The network consists of reactions of the form
\begin{equation} \notag
R_j : \sum_{i=1}^n a_{ij} X_i \rightarrow \sum_{i=1}^n b_{ij} X_i
\ \text{ with } j = 1, \dots, r,
\end{equation}

\end{definition}

\begin{example} 
\label{ex:crnt} 

Consider the reaction network $\SCR$ shown in Figure \ref{fig:ex_crnt}. It consists of the following: 
\[
\mS = \{X_1, X_2, X_3 \},
\]
\[
\mC = \{ X_1 + X_2, \ 3X_2, \ 2 X_3 \},
\]  
and   
\[
\mR = \{ X_1 + X_2 \to 2 X_3, \quad  
X_1 + X_2 \to 3X_2, \quad  
3X_2 \to 2 X_3, \quad  
2 X_3 \to X_1 + X_2 \}.
\]
\begin{figure}[!ht] 
\begin{center}
\begin{tikzpicture}
		\node (1) at (-3,0) {$\bullet$};
		\node (3) at (0,2) {$\bullet$};
		\node (2) at (3,0) {$\bullet$};
		\node [left=1pt of 1] {$X_1 + X_2$};
		\node [right=1pt of 2] {$3X_2$};
		\node (x3) at (0,2.5) {$2 X_3$};
		\draw [{->}, -{stealth}, thick, transform canvas={yshift=2pt}] (1) -- (3) node [midway, above] {};
		\draw [{->}, -{stealth}, thick, transform canvas={yshift=-0pt}] (1) -- (2) node [midway, below] {};
		\draw [{->}, -{stealth}, thick, transform canvas={xshift=1.5pt}] (2) -- (3) node [midway, right=6pt] {};
		\draw [{->}, -{stealth}, thick, transform canvas={yshift=-2pt}] (3) -- (1) node [midway, left] {};
\end{tikzpicture}
\end{center}
\caption{A reaction network $\SCR$ has three species, three complexes, and four reactions.}
\label{fig:ex_crnt}
\end{figure}
\end{example}

\begin{definition}[\cite{feinberg,guldberg1864studies}]

For a reaction network $\SCR$, each reaction $R_j$ is associated with a \textit{reaction rate constant} $k_j > 0$ for $1 \leq j \leq r$. We denote by $\bk = (k_1, \dots, k_r) \in \mathbb{R}_{>0}^{r}$ the corresponding \textit{reaction rate vector}. The \textit{mass-action system} associated with $\SCRk$ is given by
\begin{equation} \notag
\frac{d x_i (t)}{dt} = \sum_{j=1}^r k_j \prod_{l=1}^n x_l^{\alpha_{lj}} (b_{ij} - a_{ij})
\ \text{ with } i = 1, \dots, n,
\end{equation}
where $\bx (t) = (x_1(t), \dots, x_n(t))$ denotes the concentrations of species $(X_1, \dots, X_n)$ at time $t$.
\end{definition}

For example, recall the reaction network in Example~\ref{ex:crnt}. Let $\bk = (k_1, \dots, k_4)$ denote the reaction rate vector. The reaction rates are assigned as follows:
\[
X_1 + X_2 \xrightarrow{k_1} 2 X_3 \qquad  
X_1 + X_2 \xrightarrow{k_2} 3 X_2 \qquad  
3X_2 \xrightarrow{k_3} 2 X_3 \qquad  
2 X_3 \xrightarrow{k_4} X_1 + X_2
\]
Under mass-action kinetics, the corresponding dynamical system is
\begin{equation} \notag
\begin{split}
\frac{\mathrm{d}\bx}{\mathrm{d} t}
& = k_{1} x_1 x_2 \begin{pmatrix} -1 \\ -1 \\ 0 \end{pmatrix}
+ k_{2} x_1 x_2 \begin{pmatrix} -1 \\ 2 \\ 0 \end{pmatrix} 
+ k_{3} x_2^3 \begin{pmatrix} 0 \\ -3 \\2 \end{pmatrix}
+ k_{4} x_3^2 \begin{pmatrix} 1 \\ 1 \\ -2 \end{pmatrix}
\\& = \begin{pmatrix} 
- k_{1} x_1 x_2 - k_{2} x_1 x_2 + k_{4} x_3^2 
\\ -k_{1} x_1 x_2 + 2 k_{2} x_1 x_2 - 3 k_{3} x_2^3 + k_{4} x_3^2
\\ 2 k_{3} x_2^3 - 2 k_{4} x_3^2
\end{pmatrix}.
\end{split}
\end{equation}

Reaction networks provide a robust framework for molecular-level computations, as the mass-action systems they generate have polynomial right-hand sides.

\section{Modules} 
\label{sec:modules}

In this section, we develop reaction network schemes for the individual submodules that will be used to implement the \emph{linear regression} and \emph{linear interpolation} steps in Sections~\ref{sec:linear_regression} and~\ref{sec:linear_interpolation}. 
All of these modules were introduced in \cite{vasic2020crn++}, including the idea of coupling modules to oscillators. The main difference in our framework is that we employ a Hopf oscillator, as described in Section~\ref{sec:oscillation_module}.

We begin by introducing the \emph{dual-rail encoding} to accommodate negative values that may arise in the input data, weights, and biases.

\subsection{Dual-rail Encoding}

Physically, the concentration of any chemical species must be non-negative. This fundamental constraint poses a challenge for our implementation, as both the input data (data points) and the model parameters (weights and biases) may assume negative values.

To address this limitation, we adopt the dual-rail encoding paradigm \cite{10.1145/2554797.2554827}. Under this framework, each real-valued variable $\zeta \in \mathbb{R}$ is represented by the concentrations of two distinct species, $\zeta^+$ and $\zeta^-$, both of which are inherently non-negative:
\[
\zeta^+(t) \geq 0 \quad \text{and} \quad \zeta^-(t) \geq 0,
\ \text{ for all }
t \geq 0.
\]
The value of $\zeta$ is then defined as the difference between these concentrations,
\[
\zeta(t) = \zeta^+(t) - \zeta^-(t),
\ \text{ for all }
t \geq 0.
\]

In the following, we introduce several fundamental modules from \cite{vasic2020crn++}, which are utilized in the implementation of the linear regression and linear interpolation modules. We begin by introducing the notation used throughout the remainder of this paper.

\medskip

\textbf{Notation:}
Let $[X(t)]$ denote the concentration of species $X$ at time $t$. If $X$ is a catalyst, whose concentration remains constant for all $t \geq 0$, it is denoted by $[X]$. Finally, $[X]^{ss}$ denotes the steady-state concentration of $X$.

\subsection{DNA Strand Displacement Implementation}

While the modules in this section are described at the level of abstract chemical reaction networks, they can be physically implemented using toehold-mediated DNA strand displacement (TMSD) \cite{JURINOVIC2026103396}. As a representative example, consider the first reaction of the addition module (see Section \ref{sec:addition}),
\begin{equation} \notag
A \rightarrow A+C.
\end{equation}
This abstract reaction can be realized through a sequence of strand-displacement reactions involving intermediate DNA complexes. Following the construction in \cite{cimb44040119}, the corresponding reactions are
\begin{equation} \notag
% AC + A \rightarrow ACw + Af1a,
AC + A \rightarrow AC_W + AF_I
\end{equation}
and
\begin{equation} \notag
% ACf1 + Af1a \rightarrow ACf1w + f1 + A + C.
AC_F + AF_I \rightarrow AC_{FW} + F + A + C
\end{equation}
The resulting strand-displacement process is illustrated in Figure \ref{fig:dna_reactions}.

\begin{figure}[H]
    \centering
    \includegraphics[width=1.0\columnwidth]{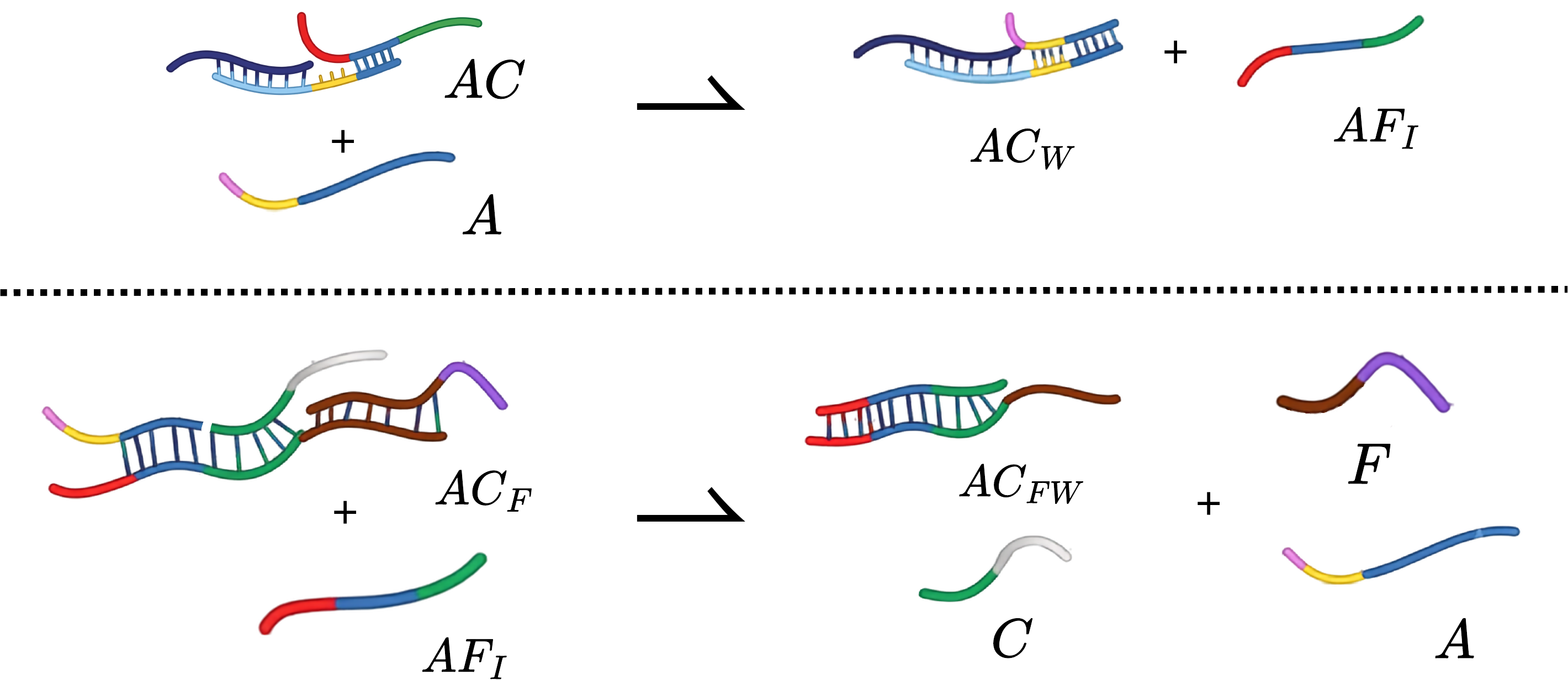}
    \caption{The figure shows the first reaction of the addition module \ref{module_addition} $A \rightarrow A + C$ using DNA strand displacement. Input strand A binds to the AC complex at the yellow toehold domain and displaces the strand $AF_I$. $AF_I$ then initiates the second reaction by hybridizing with the red toehold on $AC_F$. This  displacement breaks the complex to produce A and C along with other byproducts~\cite{cimb44040119}.}
     \label{fig:dna_reactions}
\end{figure}

At the level of the abstract reaction network, the reaction $A \rightarrow A+C$ may appear to violate mass conservation, since the product contains an additional species $C$. At the molecular level, however, the reaction represents a sequence of strand-displacement steps involving intermediate DNA complexes. The input species $A$ is temporarily consumed during the displacement process and subsequently regenerated, while the intermediate complex releases $C$. Thus, the abstract reaction describes the net computational effect of the underlying molecular mechanism rather than its detailed molecular stoichiometry.

The module constructions used in this paper build on those introduced in \cite{vasic2020crn++}, which are designed to provide \textit{convergence} and \textit{composability}. In particular, convergence ensures that the concentration dynamics of a module approach their steady state, while composability allows the input species to be preserved so that the output of one module can be used as an input to a subsequent module. The sequential execution of these modules is coordinated by the oscillator described in Section \ref{sec:oscillation_module}.

\subsection{Addition Module} \label{sec:addition}

The operation of addition is realized via the construction of a reaction network module, given by the following network:
\begin{equation} \label{module_addition}
\begin{split}
A &\overset{1}{\longrightarrow} A + C \\
B &\overset{1}{\longrightarrow} B + C \\
C &\overset{1}{\longrightarrow} \emptyset
\end{split}
\end{equation}
In this network, the input species $A$ and $B$ serve as catalysts and therefore maintain constant concentrations. Their concentrations are effectively added, with the resulting value stored in the steady-state concentration of $C$.

Suppose that $[A] = [A(0)]$ and $[B] = [B(0)]$. 
The dynamics of $C$ associated with the network \eqref{module_addition} is then given by
\begin{equation} \notag
\frac{d[C(t)]}{dt} = [A] + [B] - [C(t)].
\end{equation}
At the steady-state, where $\frac{d[C(t)]}{dt} = 0$, we obtain
\[
[C]^{ss} = [A] + [B].
\]

In subsequent sections, this module will be referred to as \texttt{Sum}.

\subsection{Multiplication Module}
\label{sec:multiplication}

The operation of multiplication is realized via the construction of a reaction network module, given by the following network:
\begin{equation} \label{module_multiplication}
\begin{split}
A + B &\overset{1}{\longrightarrow} A + B + C \\
C &\overset{1}{\longrightarrow}\emptyset
\end{split}
\end{equation}
In this network, the input species $A$ and $B$ are catalysts. Their concentrations are multiplied, with the resulting value encoded in the steady-state concentration of $C$.

Suppose that $[A] = [A(0)]$ and $[B] = [B(0)]$. 
The dynamics of $C$ associated with the network \eqref{module_multiplication} is then given by
\begin{equation} \notag
\frac{d[C(t)]}{dt} = [A] [B] - [C(t)].
\end{equation}
At the steady-state of $C$, we obtain
\[
[C]^{ss} = [A] [B].
\]

In subsequent sections, this module will be referenced as \texttt{Product}.

\begin{remark}
\label{rmk:sign_multiply}

In dual-rail encoding, each signed quantity is represented by two non-negative species. Specifically, $A$ and $B$ are encoded as $(A^+, A^-)$ and $(B^+, B^-)$, with $A=A^+ - A^-$ and $B=B^+ - B^-$. For multiplication, we have 
\[ 
A B = (A^+ - A^-) (B^+ - B^-) = A^+B^+ - A^+B^- - A^-B^+ + A^-B^-. 
\] 
Collecting the non-negative contributions yields 
\[ 
C^+ = A^+B^+ + A^-B^-, 
\qquad 
C^- = A^+B^- + A^-B^+. 
\] 
Therefore, the dual-rail outputs $C^+$ and $C^-$ can be constructed directly using four multiplication modules followed by two addition modules.
\end{remark}

\subsection{Combined Module for Multiplication and Addition} \label{sec:combined-mul-add}

In the computation of linear regression (see Section \ref{sec:linear_regression}), given the input species $\{ A_i \}^{n}_{j=1}$ and $\{ B_i \}^{n}_{j=1}$, we need to compute the sum of the products of their concentrations
\[
\sum_{i=1}^{n} [A_i] [B_i].
\]
This operation is realized via the construction of a reaction network module, given by the following network:
\begin{equation} \label{module_mul_add}
\begin{split}
A_1 + B_1 &\overset{1}\longrightarrow A_1 + B_1 + C \\
A_2 + B_2 &\overset{1}\longrightarrow A_2 + B_2 + C \\
& \vdots \qquad \qquad \vdots \\
A_n + B_n &\overset{1}\longrightarrow A_n + B_n + C \\
C &\overset{1}\longrightarrow \emptyset
\end{split}
\end{equation}
In this network, the input species $\{ A_i \}^{n}_{j=1}$ and $\{ B_i \}^{n}_{j=1}$ are catalysts. The sum of the products of their concentrations is encoded in the steady-state concentration of $C$.

Suppose that $[A_i] = [A_i (0)]$ and $[B_i] = [B_i(0)]$ for every $1 \leq i \leq n$. 
The dynamics of $C$ associated with the network \eqref{module_mul_add} is then given by
\begin{equation} \notag
\frac{d[C(t)]}{dt} = - [C(t)] + [A_1][B_1] + \cdots + [A_n][B_n].
\end{equation}
At the steady-state of $C$, we obtain
\[
[C]^{ss} = \sum_{i=1}^{n} [A_i] [B_i].
\]

\subsection{Comparison and Approximate Majority Modules}

Our module is designed to compare the concentrations of two species and set the corresponding \emph{Boolean flag species}. This functionality is implemented through a \emph{comparison module} followed by an \emph{approximate majority module} in consecutive clock cycles.

\subsubsection*{Comparison Module}

The comparison module processes the input species $X$ and $Y$ by encoding their relative concentrations, in normalized form, into two flag species, $X_{gY}$ and $X_{lY}$.
The normalization procedure is realized via the construction of a reaction network module, given by the following network:
\begin{equation} \label{module_comparison}
\begin{split}
X_{gY} + Y &\overset{1}{\longrightarrow} X_{lY} + Y \\
X_{lY} + X &\overset{1}{\longrightarrow} X_{gY} + X
\end{split}
\end{equation}
In this network, the input species $X$ and $Y$ are catalysts. 

Suppose that $[X] = [X(0)]$ and $[Y] = [Y(0)]$. 
Furthermore, we normalize the flag species by choosing their initial concentrations such that
\[
[X_{gY} (0)] + [X_{lY} (0)] = 1.
\]
The dynamics of $X_{gY}$ and $X_{lY}$ associated with the network \eqref{module_comparison} are then given by
\begin{equation} \notag
\begin{split}
\frac{d[X_{gY}(t)]}{dt} &= -[X_{gY}(t)] [Y] + [X_{lY}(t)] [X], \\
\frac{d[X_{lY}(t)]}{dt} &= [X_{gY}(t)] [Y] - [X_{lY}(t)] [X].
\end{split}
\end{equation}

Summing the two equations above shows that the quantity $[X_{gY}(t)] + [X_{lY}(t)]$ is conserved for all $t \geq 0$. At the steady-state of $X_{gY}$ and $X_{lY}$, we obtain
\begin{equation} \notag
\frac{[X_{gY}]^{ss}}{[X_{lY}]^{ss}} = \frac{[X]}{[Y]}.
\end{equation}
Since $[X_{lY}]^{ss} + [X_{gY}]^{ss} = [X_{gY} (0)] + [X_{lY} (0)] = 1$, it follows that
\begin{equation} \notag
[X_{gY}]^{ss} = \frac{[X]}{[X] + [Y]}, \qquad
[X_{lY}]^{ss} =\frac{[Y]}{[X] + [Y]}.
\end{equation}
Moreover, direct computation shows that this steady-state is globally attracting. For instance, if $[X] = 75$ and $[Y] = 25$, the flag species $X_{gY}$ and $X_{lY}$ converge to $0.75$ and $0.25$, respectively.

In subsequent sections, this module will be referenced as \texttt{Compare}.

\subsubsection*{Approximate Majority Module}

The approximate majority module \cite{cardelli2012cell} converts molecules of the species with a lower concentration into molecules of the species with a higher concentration.
After obtaining the two flag species, $X_{gY}$ and $X_{lY}$, from the comparison module, the approximate majority module acts as a \textit{threshold function}, transforming the normalized output of the comparison module into a decisive, binary flag state. This state indicates which initial input ($X$ or $Y$) had the higher concentration.
This transformation is realized via the construction of a reaction network module, given by the following network:
\begin{equation} \label{module_approximate_majority}
\begin{split}
X_{gY} + X_{lY} &\overset{1}{\longrightarrow} X_{lY} + B \\
B + X_{lY} &\overset{1}{\longrightarrow} X_{lY} + X_{lY} \\
X_{lY} + X_{gY} &\overset{1}{\longrightarrow} X_{gY} + B \\
B + X_{gY} &\overset{1}{\longrightarrow} X_{gY} + X_{gY}
\end{split}
\end{equation}
In this network, the species $B$ functions as an intermediate species.

The dynamical system associated with the network \eqref{module_approximate_majority} is then given by
\begin{equation} \notag
\begin{split}
\frac{d[X_{gY}(t)]}{dt} &= -[X_{gY}(t)] [X_{lY}(t)] + [B(t)] [X_{gY}(t)], \\
\frac{d[X_{lY}(t)]}{dt} &= -[X_{gY}(t)] [X_{lY}(t)] + [B(t)] [X_{lY}(t)], \\
\frac{d[B(t)]}{dt} &= -[B(t)] [X_{lY}(t)] - [B(t)] [X_{gY}(t)] + 2[X_{gY}(t)][X_{lY}(t)].
\end{split}
\end{equation}

Assume that the normalization step in the comparison module has been executed, so that $[X_{gY}(t)] + [X_{lY}(t)] = 1$ for all $t \geq 0$. Then $(X_{gY}(t), X_{lY}(t), B(t))$ converges to $(1, 0, 0)$ if the initial condition satisfies $[X_{gY}(0)] > [X_{lY}(0)]$, and to $(0, 1, 0)$ if $[X_{gY}(0)] < [X_{lY}(0)]$. In particular,
\begin{itemize}
    \item If $[X_{gY}(0)] > [X_{lY}(0)]$, the steady-state is $[X_{gY}]^{ss} = 1$ and $[X_{lY}]^{ss} = 0$.
    \item If $[X_{gY}(0)] < [X_{lY}(0)]$, the steady-state is $[X_{gY}]^{ss} = 0$ and $[X_{lY}]^{ss} = 1$.
\end{itemize}
We refer to \cite{vasic2020crn++} for a proof of this result.

\subsection{Subtraction Module} \label{sec:subtraction}

The subtraction module evaluates the non-negative difference between the concentrations of two input species $A$ and $B$, namely $\max(0, [A] - [B])$.
This operation is realized via the construction of a reaction network module, given by the following network:
\begin{equation} \label{module_subtraction}
\begin{split}
A &\overset{1}{\longrightarrow} A + C \\
B &\overset{1}{\longrightarrow} B + H \\
C &\overset{1}{\longrightarrow} \emptyset\\
C + H &\overset{1}{\longrightarrow} \emptyset\\
\end{split}
\end{equation}
In this network, the input species $A$ and $B$ are catalysts, and the species $H$ functions as an intermediate inhibitor.
The non-negative difference between their concentrations is encoded in the (steady-state) concentration of $C$.

Suppose that $[A] = [A(0)]$ and $[B] = [B(0)]$.
The dynamical system associated with the network \eqref{module_subtraction} is then given by
\begin{equation} \notag
\begin{split}
\frac{d[C(t)]}{dt} &= [A] - [C(t)][H(t)] - [C(t)], \\
\frac{d[H(t)]}{dt} &= [B] - [C(t)][H(t)].
\end{split}
\end{equation}
Then $C(t)$ converges to $0$ if the initial condition satisfies $[A] \leq [B]$, and to the steady-state $[A] - [B]$ if $[A] > [B]$.
This result confirms that the network implements the non-negative subtraction function $\max(0, [A] - [B])$.

In subsequent sections, this module will be referenced as \texttt{Sub}.

\subsection{Loading Module} \label{sec:loading}

The loading module transfers the concentration of species $A$ into species $B$. 
This operation is realized via the construction of a reaction network module, given by the following network:
\begin{equation} 
\label{module_loading}
\begin{split}
A &\overset{1}{\longrightarrow} A + B \\
B &\overset{1}{\longrightarrow} \emptyset
\end{split}
\end{equation}
In this network, the species $A$ is a catalyst.

Suppose that $[A] = [A(0)]$. 
The dynamics of $B$ associated with the network \eqref{module_loading} is then given by
\begin{equation} \notag
\frac{d[B (t)]}{dt} = [A] - [B].
\end{equation}
At the steady-state of $B$, we obtain
\[
[B]^{ss} = [A].
\]

\subsection{Oscillation Module}
\label{sec:oscillation_module}

For our framework, we use the Hopf oscillator as described in~\cite{choudhary2025implementationsupportvectormachines}. The use of oscillation modules as in~\cite{vasic2020crn++} (ring oscillator) is possible; however, these oscillations are not guaranteed to be perpetual. To overcome this, we use the Hopf oscillator (that converges to a limit cycle), where we assume that there exists an external measurement device that discretizes the dynamics of the Hopf oscillator into a time-signal that activates only one group of reactions at a time. 

Let $X$ and $Y$ denote the two species comprising the Hopf oscillator.
The dynamics generated by the Hopf system~\cite{kuznetsov1998elements,strogatz2024nonlinear}
is given by
\begin{equation} \label{oscillation_hopf_dynamics}
\begin{split}
\frac{d [X(t)]}{dt} &= \mu [X(t)] - [X(t)]^3 - [Y(t)]^2 [X(t)] - \omega [Y(t)], \\
\frac{d [Y(t)]}{dt} &= \mu [Y(t)] - [Y(t)]^3 - [X(t)]^2 [Y(t)] + \omega [X(t)].
\end{split}
\end{equation}
% Note that in the dual-rail encoding, \eqref{oscillation_hopf_dynamics} can be realized using mass-action kinetics.
Based on the oscillator state, an external readout device computes the following:
\begin{equation} \notag
%\label{eq:phase_to_bin}
k(t) = \left\lfloor \frac{n_{\texttt{clock}}}{2\pi} atan 2([Y(t)],[X(t)]) \right\rfloor \bmod n_{\texttt{clock}},
\end{equation}
where $n_{\texttt{clock}}$ is the number of clock slots. 
Then $O(t) \in \{0, 1\}^{n_{\texttt{clock}}}$ is defined as
\[
O_i(t) := 
\begin{cases}
1, & \text{ if } i = k(t), \\[5pt]
0, & \text{ otherwise},
\end{cases}
\quad \text{ for } i \in \{0, \dots, n_{\texttt{clock}} - 1\}.
\]
The binary nature of the oscillation molecule ensures that only certain sets of reactions get activated at a time. We exploit this in our framework to enforce sequential execution of reactions.

\section{Division Module} \label{sec:div-module}

In this section, we first recall the division module \cite{cimb44040119}. However, this module is restricted to positive values and therefore cannot directly accommodate negative inputs under the dual-rail encoding framework. 
To overcome this limitation, we introduce a generalized division module that supports arbitrary-sign real values within the dual-rail encoding representation. This generalized division module constitutes one of the main contributions of this work.
%This is one of our main contributions. We first describe the usual division module.

\subsection{Division Module for Positive Concentrations}
\label{sec:pos-div-module}

For given positive concentrations $[A]$ and $[B]$, the division module evaluates the ratio between the concentrations of the input species $A$ and $B$.
This operation is realized via the construction of a reaction network module, given by the following network:
\begin{equation} \label{module_division}
\begin{split}
A &\overset{1}{\longrightarrow} A + C \\
B + C &\overset{1}{\longrightarrow} B
\end{split}
\end{equation}
In this network, the input species $A$ and $B$ are catalysts.
The ratio between their concentrations is encoded in the steady-state concentration of $C$.

Suppose that $[A] = [A(0)]$ and $[B] = [B(0)]$.
The dynamics of $C$ associated with the network \eqref{module_division} is then given by
\begin{equation} \notag
\frac{d[C (t)]}{dt} = [A] - [B][C].
\end{equation}
At the steady-state of $C$, we obtain
\[
[C]^{ss} = \frac{[A]}{[B]}.
\]

In subsequent sections, this module will be referenced as \texttt{Abs\_div}.

\subsection{Generalized Division Module}
\label{sec:generalized-div}

Unlike multiplication between signed quantities (see Remark~\ref{rmk:sign_multiply}), division involves the ratio
\[
\frac{A}{B} = \frac{A^+ -A^-}{B^+ - B^-},
\]
which cannot be expressed as a combination of non-negative product-and-sum terms in $(A^+, A^-)$ and $(B^+,B^-)$.

To address this difficulty, we adopt the following algorithm. First, we simultaneously compute the magnitudes $|A| = |A^+ - A^-|$ and $|B| = |B^+ - B^-|$ while initializing their sign indicators. We then apply the division module $\texttt{Abs\_div} (|A|, |B|)$ to obtain $|C|$. Finally, the dual-rail outputs $(C^+, C^-)$ are assigned.

\subsubsection*{Algorithm for Generalized Division}

\begin{algorithm}[H]
\caption{Generalized Division Algorithm} \label{alg:division}
\begin{algorithmic}[1]
\Procedure{Magnitude}{$X^+, X^-$}
    \State $mag \gets \text{Sum} (\text{Sub} (X^+, X^-), \text{ Sub} (X^-, X^+))$
    \State \Return $mag$
\EndProcedure

\State $mag\_A \gets \text{Magnitude}(A^+, A^-)$
\State $mag\_B \gets \text{Magnitude}(B^+, B^-)$
\State $magC \gets \text{Abs\_div}(mag\_A, mag\_B)$

\Comment{Recovering Signs: $C_{pos}$ is 1 if $\frac{A}{B}$ is positive}
\State $A_{pos}, A_{neg} \gets \text{Compare}(A^+, A^-)$
\State $B_{pos}, B_{neg} \gets \text{Compare}(B^+, B^-)$
\State $C_{pos} \gets \text{Sum}(\text{Product}(A_{pos}, B_{pos}), \text{Product}(A_{neg}, B_{neg}))$
\State $C_{neg} \gets \text{Sum}(\text{Product}(A_{pos}, B_{neg}), \text{Product}(A_{neg}, B_{pos}))$

\Comment{Assigning Sign}
\State $C^+ \gets \text{Product}(magC, C_{pos})$
\State $C^- \gets \text{Product}(magC, 1 - C_{pos})$

\end{algorithmic}
\end{algorithm}

\begin{proposition}[Correctness of the Generalized Division Algorithm]

Let $A$ and $B$ be signed quantities with $B \neq 0$. 
Then the generalized division algorithm produces output species $C^+$ and $C^-$ satisfying
\[
C^+ - C^- = \frac{A}{B}.
\]
\end{proposition}

\begin{proof}

Suppose $A = A^+ - A^-$ and $B = B^+ - B^-$ under dual-rail encoding, with $A^+, A^-, B^+, B^- \geq 0$ and $B \neq 0$.
The algorithm proceeds in four steps.

\textbf{Step 1.}  
In lines 1-6, we compute the magnitudes $|A|$ and $|B|$.
For any dual-rail encoded variable $(X^+, X^-)$,
\[
| X | = \underbrace{\texttt{Sub}(X^+, X^-)}_{\max(0, X^+ - X^-)} 
\ + \
\underbrace{\texttt{Sub}(X^-, X^+)}_{\max(0, X^- - X^+)}.
\]
Applying this to the inputs yields the magnitudes $|A|$ and $|B|$.

\textbf{Step 2.}  
In line 7, since $|A| \ge 0$ and $|B| > 0$, we apply the division module to produce
\[
magC := |C| = \texttt{Abs\_div}(|A|, |B|) = \frac{|A|}{|B|}.
\]

\textbf{Step 3.}  
In lines 8-11, we first employ a comparison module to determine the sign of $A$ and $B$. We directly obtain boolean indicators $A_{pos}$ and $A_{neg}$ through a two-stage reaction process (with an identical parallel process for $B$). 
First, we initialize competing indicator species (initially set to equal concentration $0.5$) using $A^+$ and $A^-$:
\begin{align*}
    A^+ + A_{neg} &\overset{1}{\longrightarrow} A^+ + A_{pos} \\
    A^- + A_{pos} &\overset{1}{\longrightarrow} A^- + A_{neg}
\end{align*}
This pushes the ratio of $A_{pos}$ to $A_{neg}$ in the direction of the dominant input ($A^+$ or $A^-$). 
Next, we use an approximate majority (AM) module to amplify this difference. By introducing a transient helper species $U_A$, the AM network is given as follows:
\begin{align*}
    A_{pos} + A_{neg} &\overset{1}{\longrightarrow} A_{neg} + U_A \\
    U_A + A_{neg} &\overset{1}{\longrightarrow} 2A_{neg} \\
    A_{neg} + A_{pos} &\overset{1}{\longrightarrow} A_{pos} + U_A \\
    U_A + A_{pos} &\overset{1}{\longrightarrow} 2A_{pos}
\end{align*}
Consequently, steady-state concentration yields boolean values: 
\[
A_{pos} := \texttt{Compare}(A^+, A^-) = 
\begin{cases}
1, & \text{if } A^+ > A^- \quad(\text{$A$ is positive}), \\[5pt]
0, & \text{otherwise } \quad(\text{$A$ is non-positive}).
\end{cases}
\]
and similarly for $A_{\textbf{neg}}$ is ($1 - A_{pos}$) opposite of $A_{pos}$. An identical, parallel network yields indicators $B_{pos}$ and $B_{neg}$.

We next determine the sign of $C$ using combined multiplication and addition (Section~\ref{sec:combined-mul-add}).
The sign of $A/B$ is positive if and only if $A$ and $B$ share the same sign. Accordingly, we define
\[
C_{\mathrm{pos}} = 
\underbrace{\texttt{Product}(A_{pos}, B_{pos})}_{\texttt{Both\_positive}} 
\ + \ 
\underbrace{\texttt{Product}(A_{neg}, B_{neg})}_{\texttt{Both\_negative}},
\]
which evaluates to 1 precisely when $A$ and $B$ have the same sign, and 0 otherwise. A similar mirrored logic applies for $C_{neg}$

\textbf{Step 4.}  
The final dual-rail outputs are assigned as
\[
(C^+, C^-) = (\texttt{Product}(magC, C_{pos}), \ \texttt{Product}(magC, 1 - C_{pos})).
\]
Thus, if $C_{pos} = 1$, we have $(C^+, C^-) = (magC, 0)$; otherwise, if $C_{pos} = 0$, we have $(C^+, C^-) = (0, magC)$.

Combining the above steps, the algorithm correctly computes the dual-rail encoded quotient for arbitrary signed inputs.
\end{proof}

In subsequent sections, this module will be referenced as \texttt{Generalized Division}.
Specifically, to compute $C = \frac{A}{B}$ from dual-rail representations $(A^+, A^-)$ and $(B^+, B^-)$ of the numerator and a nonzero denominator, respectively, we denote the resulting quotient $(C^+, C^-)$ by
\[
(C^+, C^-) = \texttt{Generalized Division} \big( A^+,A^-,B^+,B^- \big).
\]

Figure \ref{fig:division_example} illustrates the evolution of species concentrations alongside the governing clock oscillations for an example computation $A = -8$ $(A^+ = 2, A^- = 10)$ and $B = 2$ $(B^+ = 4, B^- = 2)$.
\begin{figure}[H]
    \centering
    \includegraphics[width=0.8\columnwidth]{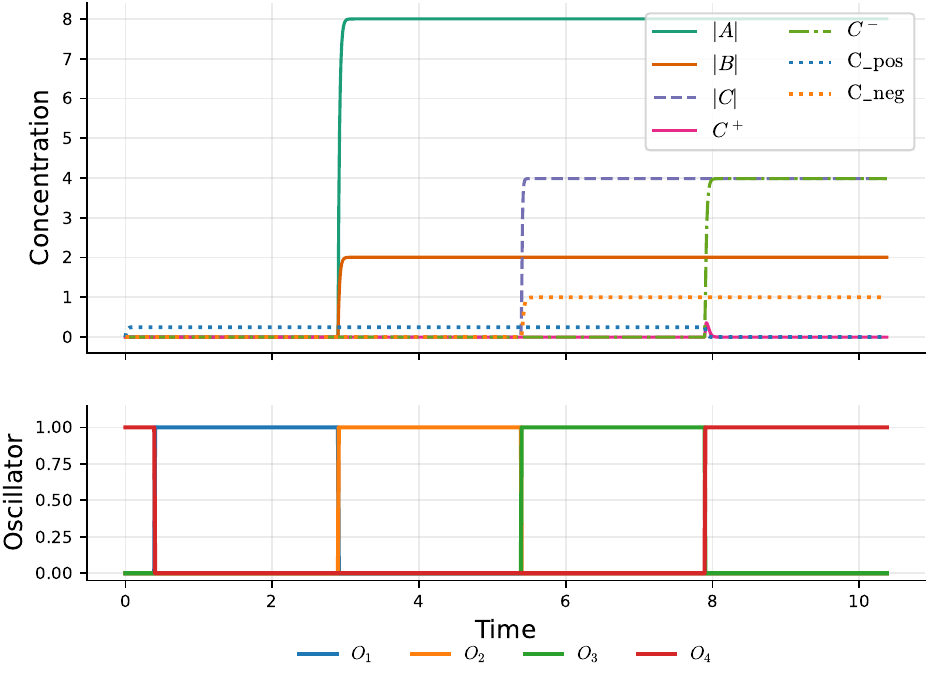}
    \caption{Generalized division module evaluating A = -8 and B = 2. (Top) Time evolution of the species concentrations. (Bottom) Sequential progression of the underlying Hopf oscillator clock phases.}
\label{fig:division_example}
\end{figure}

\paragraph{Module Oscillation Count:}
The overall module execution time is determined by the longest sequential pathway, measured in terms of oscillations required for module completion. By parallelizing independent operations, the reaction completes in 4 oscillations.
\begin{enumerate}
    \item \textbf{Oscillation 1: } Partial magnitude subtractions ($A_{sub}^+, A_{sub}^-$) and initial sign comparisons ($A_{pos}, A_{neg}$) are evaluated in parallel (1 oscillation).
    
    \item \textbf{Oscillation 2: } Partial magnitudes are summed into magnitudes (mag\_A and mag\_B), while the Approximate Majority module makes the sign Boolean exact Boolean states (1 oscillation).
    
    \item \textbf{Oscillation 3: } The $\texttt{Abs\_div}$ module computes the magnitude ratio ($\lvert C \rvert$) while the sign terms ($C_{pos}, C_{neg}$) are simultaneously evaluated using combined multiplication and addition. (1 oscillation).
    
    \item \textbf{Oscillation 4: } Finally, we assign magC to $C^+$ and $C^-$  based on evaluated sign variable (1 oscillation).
\end{enumerate}
The total execution time for the generalized division module is 4 oscillations.

\section{Linear Regression via CRN}
\label{sec:linear_regression}

Linear regression is a supervised learning method in which we seek a function that best fits a set of data points.
In univariate linear regression, given data pairs $\{ (x_i, y_i) \}^n_{i=1}$ with $x_i, y_i \in \mathbb{R}$, the goal is to determine parameters $w$ and $b$ (the slope and intercept) such that $y \approx wx + b$ for all observations.
In multivariate linear regression with $d$ features, the data take the form $\{ (\mathbf{x}_i, y_i) \}^n_{i=1}$ with $\mathbf{x}_i \in \mathbb{R}^d$ and $y_i \in \mathbb{R}$. The objective is to determine a parameter vector $\mathbf{W} = (W_1, \dots, W_d)$ and intercept $b$ such that $y \approx \mathbf{W} \cdot x + b$, thereby fitting the data.

In Section \ref{sec:uni_regression}, we apply the CRN method to implement univariate linear regression. In Section \ref{sec:multi_regression}, we extend the CRN scheme to train a multivariate linear regression model using gradient descent.

By a slight abuse of notation, throughout this section, we use the same symbols (e.g., $x_i^+$) to denote both the species and their corresponding concentrations.

\subsection{Univariate Linear Regression Using CRN}
\label{sec:uni_regression}

In this section, we construct a chemical reaction network (CRN) that computes the solution to univariate linear regression through a sequence of oscillation-controlled reactions.
Given data pairs $\{ (x_i, y_i) \}^n_{i=1}$ with $x_i, y_i \in \mathbb{R}$, we seek parameters $(w, b)$ such that $y \approx wx + b$. This is typically achieved by minimizing the mean squared error (MSE), defined as
\[
\mathrm{MSE}(w,b) = \sum_{i=1}^{n} (y_i - (w x_i + b))^2.
\]
From \cite{kenney1962linear}, the MSE is minimized by the following optimal slope $\widehat{w}$ and intercept $\widehat{b}$:
\begin{equation} \label{eq:univariate_sol}
\widehat{w}
= \frac{n ( \sum_{i=1}^{n} x_i y_i ) - ( \sum_{i=1}^{n} x_i ) ( \sum_{i=1}^{n} y_i )}{n ( \sum_{i=1}^{n} x_i^2 ) - ( \sum_{i=1}^{n} x_i )^2},
\qquad
\widehat{b}
= \frac{1}{n} ( \sum_{i=1}^{n} y_i ) - \widehat{w} \ \frac{1}{n} ( \sum_{i=1}^{n} x_i ).
\end{equation}

We now show that the optimal slope $\widehat{w}$ and intercept $\widehat{b}$ in \eqref{eq:univariate_sol} can be computed using the modules described in Sections~\ref{sec:modules} and~\ref{sec:div-module}. To this end, we first introduce the following notation:
\begin{equation} \label{def:notation_1}
\Sigma_x = \sum_{i=1}^{n} x_i, \quad
\Sigma_y = \sum_{i=1}^{n} y_i, \quad
\Sigma_{xx} = \sum_{i=1}^{n} x_i^2, \quad
\Sigma_{xy} = \sum_{i=1}^{n} x_i y_i.
\end{equation}

\paragraph{Implementation}

We describe a chemical reaction network (CRN) that computes the least-squares estimators using dual-rail encoding.

\begin{enumerate}
\item \textbf{Dual-rail encoding.}
Each input is represented in dual-rail form as
\[
x_i = x_i^+ - x_i^-, 
\qquad 
y_i = y_i^+ - y_i^-.
\]

\item \textbf{Computation of linear and bilinear sums.}
We compute the aggregate quantities \eqref{def:notation_1} in dual-rail form. Here, we show the implementation of $\Sigma_{xy}$ and $\Sigma_x$.

\smallskip

$(a)$ For $\Sigma_{xy}$, we rewrite it as
\[
\Sigma_{xy} = \underbrace{\sum_{i=1}^n (x_i^+ y_i^+ + x_i^- y_i^-)}_{:= \Sigma_{xy}^+} \ - \ \underbrace{\sum_{i=1}^n (x_i^+ y_i^- + x_i^- y_i^+)}_{:= \Sigma_{xy}^-}.
\]
To compute the positive component $\Sigma_{xy}^+$, we implement the combined module for multiplication and addition (Section~\ref{sec:combined-mul-add}) via the following network:
\begin{equation} \notag
\begin{split}
x_1^+ + y_1^+ + O_1 & \overset{1}\longrightarrow x_1^+ + y_1^+ + \Sigma_{xy}^+ +O_1 
\\ x_1^- + y_1^- + O_1 &\overset{1}\longrightarrow x_1^- + y_1^- + \Sigma_{xy}^+ + O_1 
\\& \vdots 
\\ x_n^+ + y_n^+ + O_1 &\overset{1}\longrightarrow x_n^+ + y_n^+ + \Sigma_{xy}^+ + O_1
\\ x_n^- + y_n^- + O_1 &\overset{1}\longrightarrow x_n^- + y_n^- + \Sigma_{xy}^+ + O_1
\\ \Sigma_{xy}^+ + O_1&\overset{1}\longrightarrow O_1
\end{split}
\end{equation}
Here, the oscillating species $O_1$ induces a time-scale separation corresponding to the active computation phase of the module.
In this network, the input species $\{ x_i^{\pm} \}^{n}_{j=1}$, $\{ y_i^{\pm} \}^{n}_{j=1}$, and $O_1$ are catalysts.
Thus, we derive
\[
\frac{d[\Sigma_{xy}^+ (t)]}{dt} = \big( [x_1^+][y_1^+] + \cdots + [x_n^+][y_n^+] + [x_1^-][y_1^-] + \cdots +[x_n^-][y_n^-] - [\Sigma_{xy}^+ (t)] \big) [O_1].
\]
Note that $[O_1] > 0$ during the computation phase. Hence, the steady-state of $\Sigma_{xy}^+$ satisfies that
\[
[\Sigma_{xy}^+]^{ss} = \sum_{i=1}^n ([x_i^+] [y_i^+] + [x_i^-] [y_i^-]). 
\]
The network for the negative component $\Sigma_{xy}^-$ is constructed analogously and in parallel.

\smallskip

$(b)$ For $\Sigma_{x}$, we implement the addition module (Section \ref{sec:addition}) during the $O_1$ phase via the following network:
\begin{equation} \notag
\begin{split}
x_1^+ + O_1 & \overset{1}\longrightarrow x_1^+ + \Sigma_x^+ +O_1\\
x_2^+ + O_1 & \overset{1}\longrightarrow x_2^+ + \Sigma_x^+ +O_1
\\& \vdots 
\\ x_n^+ + O_1 & \overset{1}\longrightarrow x_n^+ + \Sigma_x^+ +O_1
\end{split}
\end{equation}
In this network, the input species $\{ x_i^{+} \}^{n}_{j=1}$ and $O_1$ are catalysts. Thus, we derive
\[
\frac{d[\Sigma_{x}^+ (t)]}{dt} = \bigl([x_1^+]+ [x_2^+]+ \dots +[x_n^+]] \bigr) [O_1].
\]
Since $[O_1] > 0$ during the computation phase, the steady-state of $\Sigma_{x}^+$ follows
\[
[\Sigma_{x}^+]^{ss} = \sum_{i=1}^n [x_i^+] . 
\]
The networks corresponding to the negative component $\Sigma_x^-$, as well as those for $\Sigma_y$ and $\Sigma_{xx}$, are constructed analogously and operate in parallel.

\item \textbf{Computation of products of aggregate sums.}

We now compute the products of aggregate quantities:
\[
P_1 := n \Sigma_{xy},
\qquad 
P_2 := \Sigma_x \Sigma_y,
\qquad 
P_3 := n \Sigma_{xx},
\qquad 
P_4 := (\Sigma_x)^2,
\]
in dual-rail form.
Here, we show the implementation of $P_2 = \Sigma_x \Sigma_y$.
We rewrite $P_2$ as
\[
P_2 = (\Sigma_x^{+} - \Sigma_x^{-})(\Sigma_y^{+} - \Sigma_y^{-})
= \underbrace{(\Sigma_x^{+}\Sigma_y^{+} + \Sigma_x^{-}\Sigma_y^{-} )}_{:= P_2^{+}}
- \underbrace{(\Sigma_x^{+}\Sigma_y^{-} + \Sigma_x^{-}\Sigma_y^{+} )}_{:= P_2^{-}}.
\]

For the positive component $P_2^{+}$, we again implement the combined module for multiplication and addition (Section~\ref{sec:combined-mul-add}) via the following network:
\begin{equation} \notag
\begin{split}
\Sigma_x^+ + \Sigma_y^+ + O_2 &\overset{1}{\longrightarrow} \Sigma_x^+ + \Sigma_y^+ + P_2^+ + O_2 
\\ \Sigma_x^- + \Sigma_y^- + O_2 &\overset{1}{\longrightarrow} \Sigma_x^- + \Sigma_y^- + P_2^+ + O_2 
\\ \Sigma_x^+ + \Sigma_y^- + O_2 &\overset{1}{\longrightarrow} \Sigma_x^+ + \Sigma_y^- + P_2^- + O_2 
\\ \Sigma_x^- + \Sigma_y^+ + O_2 &\overset{1}{\longrightarrow} \Sigma_x^- + \Sigma_y^+ + P_2^- + O_2 
\\ P_2^+ + O_2 &\overset{1}{\longrightarrow} O_2 
\\ P_2^- + O_2 &\overset{1}{\longrightarrow} O_2
\end{split}
\end{equation}
Under mass-action kinetics, the above network ensures that, during the $O_2$-active phase, the steady-state of $P_2^+$ is
\[
[P_2^+]^{ss} = \sum_{i=1}^n [x_i^+] \sum_{i=1}^n [y_i^+] + \sum_{i=1}^n [x_i^-] \sum_{i=1}^n [y_i^-]. 
\]
The networks corresponding to the negative component $P_2^-$, as well as those for $P_1$, $P_3$, and $P_4$, are constructed analogously and operate during the $O_2$-active phase.

\item \textbf{Computation of the optimal slope $\widehat{w}$.}

From \eqref{eq:univariate_sol}, we have 
\begin{equation} \label{eq:hatw}
\widehat{w} = \frac{ P_1 - P_2}{P_3 - P_4}.
\end{equation}
Here $P_1$, $P_2$, $P_3$ and $P_4$ are the product terms computed in the previous step.
Moreover, the numerator and denominator on the right-hand side follow that
\begin{equation} \notag
\begin{split}
P_1 - P_2 & = \underbrace{(P_1^+ + P_2^-)}_{:= Num^+} \ - \ \underbrace{(P_1^- + P_2^+)}_{:= Num^-},
\quad
P_3 - P_4  = \underbrace{(P_3^+ + P_4^-)}_{:= Den^+} \ - \ \underbrace{(P_3^- + P_4^+)}_{:= Den^-}.
\end{split}
\end{equation}

For $Num^+ = P_1^+ + P_2^-$, we implement the addition module (Section \ref{sec:addition}) during the $O_3$ phase via the following network:
\begin{equation} \notag
\begin{split}
P_1^+ + O_3 &\overset{1}\longrightarrow P_1^+ + Num^+ + O_3 
\\ P_2^- + O_3 &\overset{1}\longrightarrow P_2^- + Num^+ + O_3 
\\ Num^+ + O_3 &\overset{1}\longrightarrow O_3
\end{split}
\end{equation}
Here, we introduce the oscillating species $O_3$ to ensure that this operation is activated only after the $O_1$-phase.
By direct computation, during the $O_3$-active phase, the steady-state of $Num^+$ is
\[
[Num^+]^{ss} = [P_1^+] + [P_2^-].
\]
The network for the negative component $Num^-$, as well as those for $Den^+$ and $Den^-$, are constructed analogously and operate during the $O_3$-active phase.

After constructing the dual-rail forms of the numerator and denominator in \eqref{eq:hatw}, we apply the generalized division module (Section~\ref{sec:generalized-div}) to obtain that
\[
(\widehat{w}^+, \widehat{w}^-)
= \texttt{Generalized Division} \big( \mathrm{Num}^+,\mathrm{Num}^-,\mathrm{Den}^+,\mathrm{Den}^- \big).
\]

\item \textbf{Computation of the optimal intercept $\widehat{b}$.}

From \eqref{eq:univariate_sol}, the optimal intercept $\widehat{b}$ follows that
\[
\widehat{b} = \underbrace{ \frac{1}{n^2} ( n \Sigma_y^{+} + \widehat{w}^- \Sigma_x^{+} + \widehat{w}^+ \Sigma_x^{-}) }_{:= \widehat{b}^{+}}
\ - \ \underbrace{ \frac{1}{n^2} ( n \Sigma_y^{-} + \widehat{w}^+ \Sigma_x^{+} + \widehat{w}^- \Sigma_x^{-} )}_{:= \widehat{b}^{-}}.
\]

For $\widehat{b}^{+} = \frac{1}{n^2} ( n \Sigma_y^{+} + \widehat{w}^- \Sigma_x^{+} + \widehat{w}^+ \Sigma_x^{-})$, the CRN realizes the expression in two steps.

\begin{figure}[H]
    \centering
    \includegraphics[width=0.5\columnwidth]{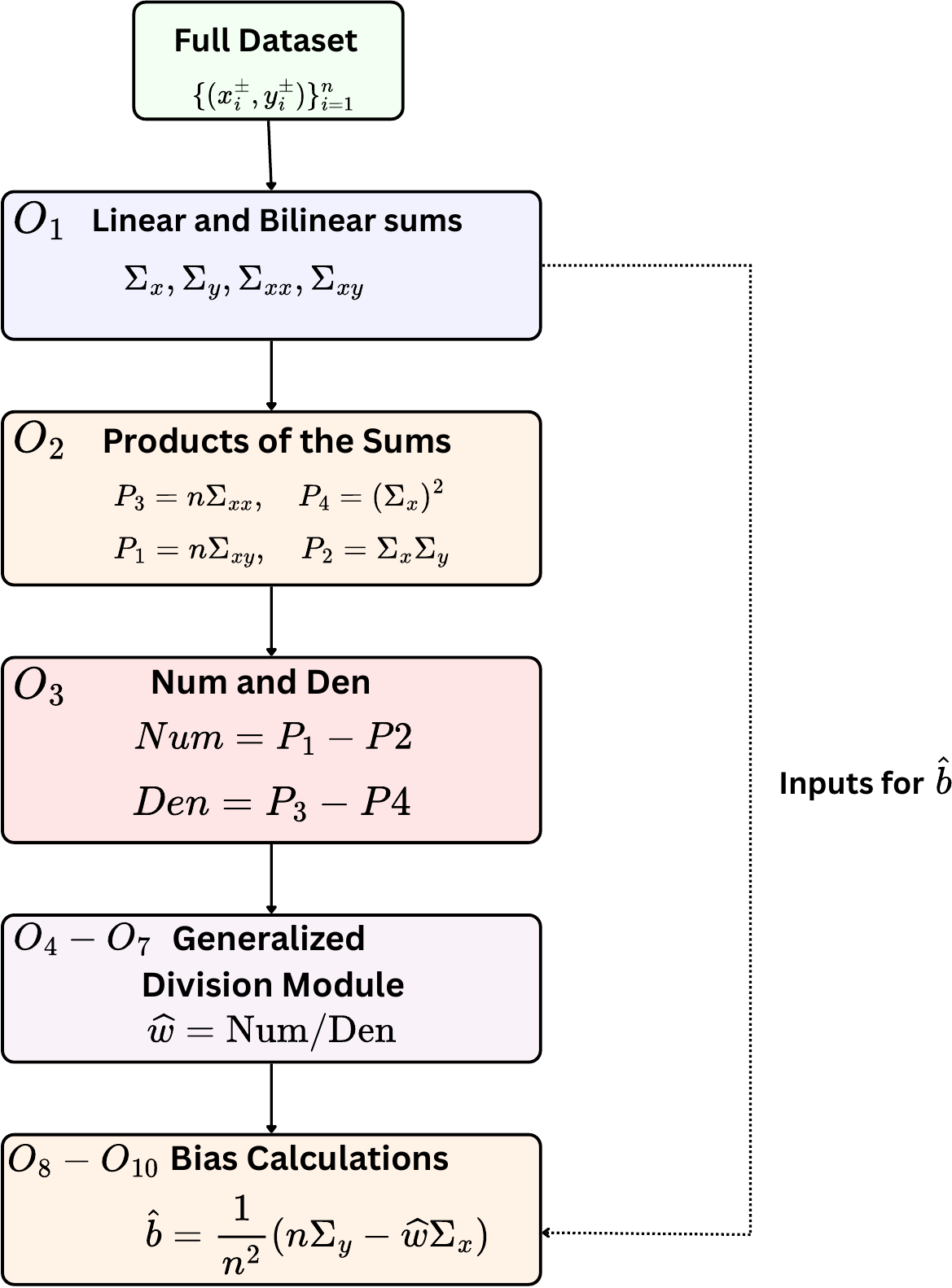}
    % \caption{Chemical Reaction Network Flow for Univariate Linear Regression.}
    \caption{Flowchart illustrating the reaction-network implementation of univariate linear regression.}
    \label{fig:univariat_schematic}
\end{figure}

First, during the $O_5$-active phase, we implement the combined module for multiplication and addition (Section~\ref{sec:combined-mul-add}) to compute 
\[
\widehat{w}^- \Sigma_x^{+} + \widehat{w}^+ \Sigma_x^{-}.
\]
Second, during the $O_7$-active phase, we implement the addition module (Section \ref{sec:addition}) to realize 
\[
n \Sigma_y^{+} + \widehat{w}^- \Sigma_x^{+} + \widehat{w}^+ \Sigma_x^{-}.
\]

The networks for the negative component $\widehat{b}^-$ are constructed analogously and operate during the $O_8$ and $O_9$-active phase.
Figure~\ref{fig:univariat_schematic}  represents a flowchart for univariate linear regression.
\end{enumerate}

We validate the construction through simulations in Example~\ref{ex:univarite}.

\begin{example} \label{ex:univarite}

Figure~\ref{fig:univariate_plot} compares the proposed CRN-based linear regression method with the least-squares fit computed in Python. To this end, we generate 40 data points. The input features ($x$) are sampled uniformly, and the target values ($y$) are generated using relationship $y  = 2.5x - 1 + \varepsilon$, where $\varepsilon$ is Gaussian noise.

\begin{figure}[H]
    \centering
    \includegraphics[width=0.7\linewidth]{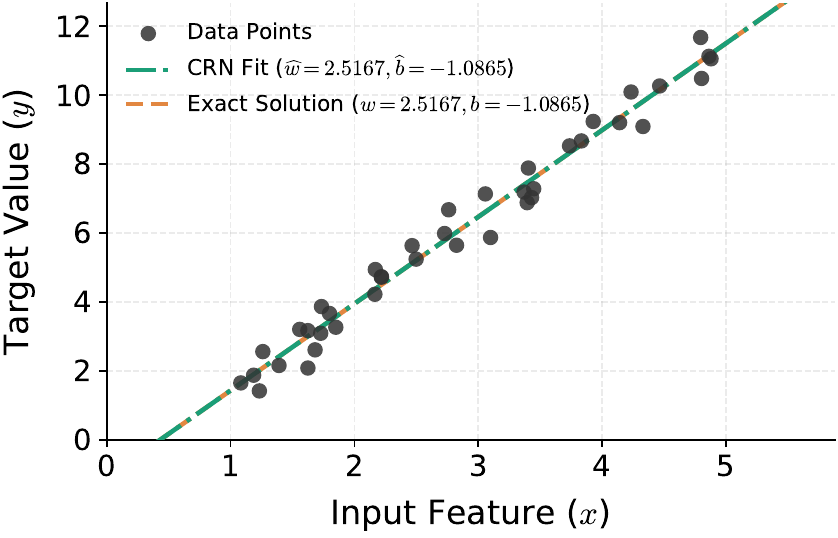}
    \caption{The scatter points represent 40 synthetic data points. The CRN accurately computes the optimal slope ($\widehat{w}=2.5167$) and intercept ($\widehat{b}=-1.0865$), shown as the green dashed line. The resulting regression line overlaps with the Python solution (orange dashed line).}
    \label{fig:univariate_plot}
\end{figure}
\end{example}

%}

\subsection{Multivariate Linear Regression Using Gradient Descent}
\label{sec:multi_regression}

Motivated by \cite{vasic2020crn++} and employed in the context of support vector machines in \cite[Section~6]{choudhary2025implementationsupportvectormachines}, we utilize the \emph{assignment module} to partition the input data into batches and load them in parallel. We do not repeat the construction here and instead refer the reader to \cite{choudhary2025implementationsupportvectormachines} for details.

To implement multivariate linear regression within the CRN framework, we adopt a gradient descent scheme, as a closed-form solution is not directly realizable. Specifically, the parameters are iteratively updated to minimize the mean squared error (MSE).
Given data pairs $\{ (\mathbf{x}_i, y_i) \}^n_{i=1}$ with $\mathbf{x}_i \in \mathbb{R}^d$ and $y_i \in \mathbb{R}$, we aim to determine a parameter vector $\mathbf{W} = (W_{1}, \ldots, W_{d}) \in \mathbb{R}^d$ and intercept $b \in \mathbb{R}$ that minimize
\[
\mathrm{MSE} (\mathbf{W}, b) = \frac{1}{n} \sum_{i=1}^{n} \big( (\mathbf{W} \cdot \mathbf{x}_i + b) - y_i \big)^2.
\]
The gradient descent updates are given by
\begin{equation} \notag
\mathbf{W} \leftarrow \mathbf{W} - \eta \frac{\partial MSE}{\partial W}, \qquad b \leftarrow b - \eta \frac{\partial MSE}{\partial b},
\end{equation}
where $\eta > 0$ is the learning rate, and the gradients take the form
\begin{equation} \label{eq:multivariate_sol}
\frac{\partial MSE}{\partial W} = \frac{2}{n} \sum_{i=1}^{n} (\mathbf{W} \cdot \mathbf{x}_i + b - y_i) \mathbf{x}_i, \qquad 
\frac{\partial MSE}{\partial b} = \frac{2}{n} \sum_{i=1}^{n} (\mathbf{W} \cdot \mathbf{x}_i + b - y_i).
\end{equation}
We now show that the MSE gradients $\frac{\partial MSE}{\partial W}$ and $\frac{\partial MSE}{\partial b}$ in \eqref{eq:multivariate_sol} can be computed using the modules described in Sections~\ref{sec:modules} and~\ref{sec:div-module}. 

For the implementation, we partition the $n$ input data points into $k$ batches, each of size $\tilde{p}$, so that $n = k \tilde{p}$. Instead of computing the gradients over all $n$ data points in \eqref{eq:multivariate_sol}, we evaluate them within each batch using only $\tilde{p}$ data points. Accordingly, the gradient expressions are realized in a batch-wise manner, with $n$ replaced by $\tilde{p}$ in each computation.

\paragraph{Implementation.}
We construct a chemical reaction network (CRN) that implements multivariate linear regression training via gradient descent.
%{\color{red}
Throughout this implementation, subscripts $i \in \{ 1,\dots,d \}$ denote feature dimensions, while superscripts $\ell \in \{ 1,\dots,\tilde{p} \}$ denote the index of a data point within a batch. Following the assignment module architecture in \cite{choudhary2025implementationsupportvectormachines}, data within each dataset, can be processed in parallel lanes provided the computations are independent, as shown in Figure~\ref{fig:Multivariate_schematic}. In Steps 2, 3, and the initial part of Step 4, we present the reaction network for a single data point and temporarily suppress the superscript notation for simplicity.

\begin{enumerate}

\item \textbf{Initialization.}
We initialize the parameter vector $\mathbf{W} = (1, \ldots, 1) \in \mathbb{R}^d$ and the intercept $b = 1$.

\item \textbf{Feedforward computation.}  
Here, we detail the computations required for a single data pair $(\mathbf{x}, y ) \in \mathbb{R}^d \times \mathbb{R}$. 
The weights $\mathbf{W}$ and bias $b$ are shared globally across all lanes corresponding to the data set $\{ (\mathbf{x}^{(l)}, y^{(l)}) \}^{\tilde{p}}_{l=1}$.

For a given feature vector $\mathbf{x} = (x_1, \dots, x_d) \in \mathbb{R}^d$ and target value $y \in \mathbb{R}$, and under the current parameter vector $\mathbf{W} = (W_{1}, \ldots, W_{d}) \in \mathbb{R}^d$ and intercept $b \in \mathbb{R}$, we evaluate
\[
P = \mathbf{W} \cdot \mathbf{x} + b.
\]
We compute $P = P^+ - P^-$ in dual-rail form. 
To compute the positive component $P^+$, we implement the combined module for multiplication and addition (Section~\ref{sec:combined-mul-add}) via the following network:
\begin{align*}
x_{1}^+ + W_{1}^+ + O_1 & \xrightarrow{1} x_{1}^+ + W_{1}^+ + P^+ + O_1 \\
x_{1}^- + W_{1}^- + O_1 & \xrightarrow{1} x_{1}^- + W_{1}^- + P^+ + O_1 \\
& \vdots \\
x_{d}^+ + W_{d}^+ + O_1 & \xrightarrow{1} x_{d}^+ + W_{d}^+ + P^+ + O_1 \\
x_{d}^- + W_{d}^- + O_1 & \xrightarrow{1} x_{d}^- + W_{d}^- + P^+ + O_1 \\
b^+ + O_1 & \xrightarrow{1} b^+ + P^+ + O_1 \\
P^+ + O_1 & \xrightarrow{1} O_1 
\end{align*}
Under mass-action kinetics, the above network ensures that, during the $O_1$-active phase, the steady-state of $P^+$ is
\[
[P^+]^{ss} = \sum_{i=1}^d ([x_i^+] [W_i^+] + [x_i^-] [W_i^-] ) + [b^+]. 
\]
The network corresponding to the negative component $P^-$ is constructed analogously and operates during the $O_1$-active phase.

\item \textbf{Error computation.}  
The prediction error $E = y - P$ is computed in dual-rail form:
\[
E = \underbrace{(y^+ + P^-)}_{E^+} \ - \ \underbrace{(y^- + P^+)}_{E^-}.
\]
To compute the positive component $E^+$, we implement the addition module (Section~\ref{sec:addition}) via the following network:
\begin{align*}
P^- + O_2 &\xrightarrow{1} P^- + E^+ + O_2 \\
y^+ + O_2 &\xrightarrow{1} y^+ + E^+ + O_2 \\
E^+ + O_2 &\xrightarrow{1} O_2
\end{align*}
Under mass-action kinetics, the above network ensures that, during the $O_2$-active phase, the steady-state of $E^+$ is
\[
[E^+]^{ss} = [y^+] + [P^-]. 
\]
The network corresponding to the negative component $E^-$ is constructed analogously and operates during the $O_2$-active phase.

\item \textbf{Gradient accumulation.} 
First, we compute the gradient update. Because our error species is defined as $E = y-P$, it gives a negative gradient, which cancels the subtraction of the gradient. 
We express
\[
\mathbf{R} = (R_1, \ldots, R_d) = \mathbf{W} + \frac{2 \eta}{\tilde{p}} E \mathbf{x},
\]
in dual-rail form. For $i = 1, \ldots, d$,
\[
R_i = \underbrace{ \Big( W_i^+ + \frac{2 \eta}{\tilde{p}} (E^+ x_i^+ + E^- x_i^-) \Big) }_{R_i^+} \ - \ \underbrace{ \Big( W_i^- + \frac{2 \eta}{\tilde{p}} (E^+ x_i^- + E^- x_i^+) \Big) }_{R_i^-}.
\]

Let $\eta' = \frac{2 \eta}{\tilde{p}}$ denote the effective batch learning rate.
To compute $(R_i^+, R_i^-)$, we implement the combined module for multiplication and addition (Section~\ref{sec:combined-mul-add}) via the following network:
\begin{equation} \notag
\begin{split}
E^+ + x_{i}^+ + O_3 &\xrightarrow{\eta'} E^+ + x_{i}^+ + R_{i}^+ + O_3 \\
E^- + x_{i}^- + O_3 &\xrightarrow{\eta'} E^- + x_{i}^- + R_{i}^+ + O_3 \\
E^+ + x_{i}^- + O_3 &\xrightarrow{\eta'} E^+ + x_{i}^- + R_{i}^- + O_3 \\
E^- + x_{i}^+ + O_3 &\xrightarrow{\eta'} E^- + x_{i}^+ + R_{i}^- + O_3 \\
W_{i}^+ + O_3 &\xrightarrow{1} W_{i}^+ + R_{i}^+ + O_3 \\
W_{i}^- + O_3 &\xrightarrow{1} W_{i}^- + R_{i}^- + O_3 \\
R_{i}^+ + O_3 &\xrightarrow{1} O_3 \\
R_{i}^- + O_3 &\xrightarrow{1} O_3
\end{split}
\end{equation}
Under mass-action kinetics, the above network ensures that, during the $O_3$-active phase, the steady-states of $(R_i^+, R_i^-)$ are
\begin{equation} \notag
\begin{split}
[R_i^+]^{ss} & = [W_{i}^+] + \eta' \big( [E^{+}][x_{i}^{+}] + [E^{-}][x_{i}^{-}] \big),
\\ [R_i^-]^{ss} & = [W_{i}^-] + \eta' \big( [E^{+}][x_{i}^{-}] + [E^{-}][x_{i}^{+}] \big).
\end{split}
\end{equation}

Note that the above corresponds to the gradient update for a single data pair $(\mathbf{x}, y )$.
To implement one batch dataset $\{ (\mathbf{x}^{(l)}, y^{(l)}) \}^{\tilde{p}}_{l=1}$, we incorporate all $\tilde{p}$ data points into the networks described in the previous steps and obtain the steady-states as follows: for $i = 1, \ldots, d$,
\begin{equation} \notag
\begin{split}
[R_i^+]^{ss} & = [W_{i}^+] + \eta' \sum_{l=1}^{\tilde{p}} \big( [E^{+}]^{(l)} [x_{i}^{+}]^{(l)} + [E^{-}]^{(l)} [x_{i}^{-}]^{(l)} \big),
\\ [R_i^-]^{ss} & = [W_{i}^-] + \eta' \sum_{l=1}^{\tilde{p}} \big( [E^{+}]^{(l)} [x_{i}^{-}]^{(l)} + [E^{-}]^{(l)} [x_{i}^{+}]^{(l)} \big),
\end{split}
\end{equation}
where the superscript $^{(l)}$ denotes quantities corresponding to the $l$-th data point in the batch.

Second, we compute the intercept update $S = b + \frac{2 \eta}{\tilde{p}} E$ in dual-rail form:
\[
S = \underbrace{\big( b^+ + \frac{2 \eta}{\tilde{p}} E^+ \big)}_{S^+} \ - \ \underbrace{\big( b^- + \frac{2 \eta}{\tilde{p}} E^- \big)}_{S^-}.
\]
Using similar networks from previous steps, we obtain
\begin{equation} \notag
\begin{split}
[S^+]^{ss} = [b^+] + \frac{2 \eta}{\tilde{p}} \sum_{l=1}^{\tilde{p}} [E^{+}]^{(l)}, 
\qquad
[S^-]^{ss} = [b^-] + \frac{2 \eta}{\tilde{p}} \sum_{l=1}^{\tilde{p}} [E^{-}]^{(l)}.
\end{split}
\end{equation}
We omit the details as the construction is analogous to the gradient update.

\item \textbf{Parameter update.} 
We transfer the updated parameter values $\mathbf{R}$ and $S$ to $\mathbf{W}$ and $b$, respectively. To load $\mathbf{R}$ into $\mathbf{W}$, we implement the loading module (Section~\ref{sec:loading}) via the following networks: for $i = 1, \ldots, d$,
\begin{align*}
R_{i}^+ + O_4 &\xrightarrow{1} R_{i}^+ + W_{i}^+ + O_4 \\
R_{i}^- + O_4 &\xrightarrow{1} R_{i}^- + W_{i}^- + O_4 \\
W_{i}^+ + O_4 &\xrightarrow{1} O_4 \\
W_{i}^- + O_4 &\xrightarrow{1} O_4
\end{align*}

At steady-state, we obtain that for $i = 1, \ldots, d$,
\begin{equation} \notag
\begin{split}
[W_{i}^{+}]^{ss} = [R_{i}^{+}],
\qquad
[W_{i}^{-}]^{ss} = [R_{i}^{-}].
\end{split}
\end{equation}
The loading of $S$ into $b$ is realized by an analogous network, and we therefore omit the details.
\end{enumerate}

\begin{figure}[H]
    \centering
    \includegraphics[width=\columnwidth]{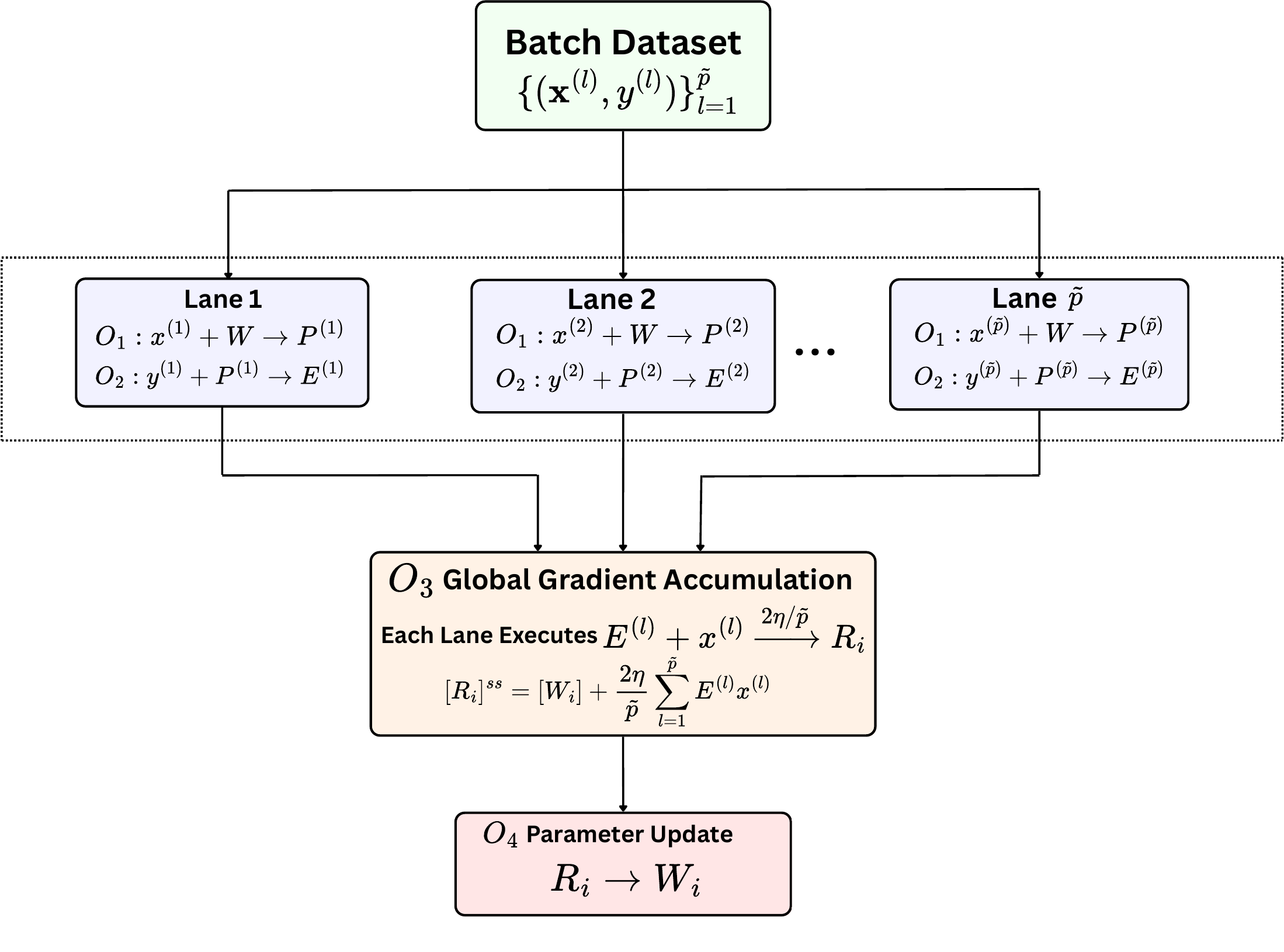}
    \caption{Flowchart illustrating the reaction-network implementation of multivariate linear regression using batch gradient descent.}
    \label{fig:Multivariate_schematic}
\end{figure}

%{\color{red}

We validate our results through simulations in Example~\ref{ex:multi_variate}.

\begin{example} \label{ex:multi_variate}

This example illustrates gradient descent for 2D linear regression using chemical reaction networks. For the plots, we generate a synthetic dataset consisting of $N=20$ points with the relationship $y = 1.5x_1 - 0.8x_2 + 0.5 + \varepsilon$, where $\varepsilon$ is Gaussian noise. We observe the convergence of the weights and bias in Figure~\ref{fig:multi_variate_plots}.

\begin{figure}[H]
    \centering
    \includegraphics[width=1\columnwidth]{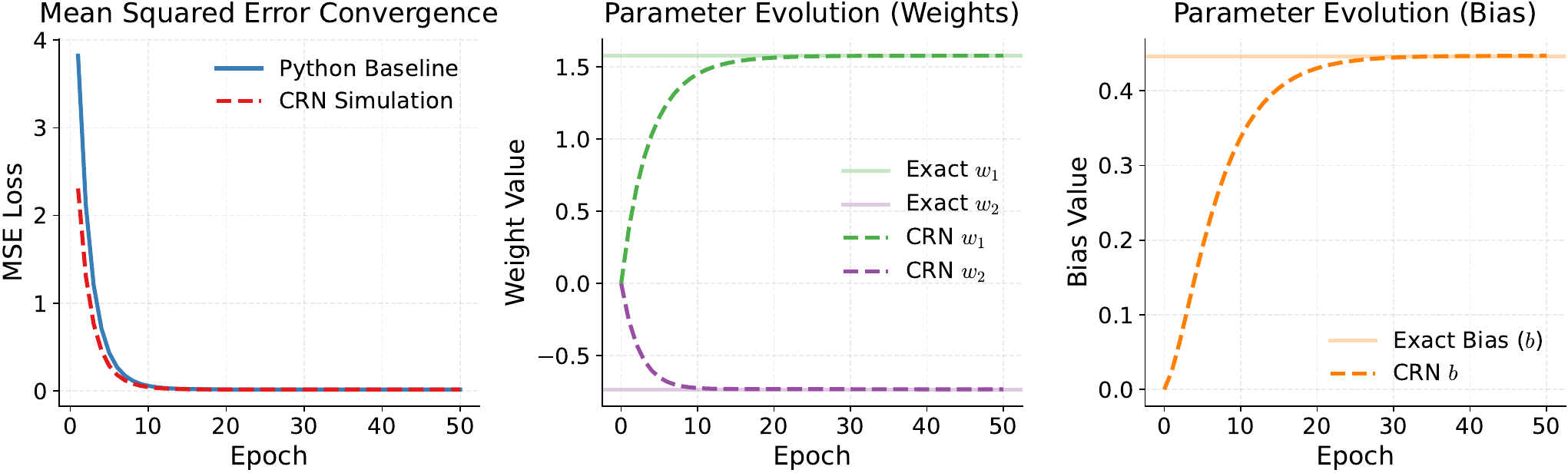}
    \caption{The plots show the convergence of CRN-based gradient descent linear regression compared with the Python baseline over 50 epochs with a learning rate of 0.1. The left, center, and right panels correspond to the MSE loss, weights, and bias, respectively. As shown, the MSE for both the CRN simulation and the Python simulation approaches zero. The CRN accurately recovers the target weights and bias.}
    \label{fig:multi_variate_plots}
\end{figure}
\end{example}

%}

\section{Linear Interpolation}
\label{sec:linear_interpolation}

Linear interpolation is a method used to estimate unknown values between two known data points by assuming a linear relationship. Given two data pairs $(x_0, y_0)$ and $(x_1, y_1)$ with $x_0, x_1, y_0, y_1 \in \mathbb{R}$, the goal is to compute the value of $y$ corresponding to a given $x$. The interpolation formula is
\begin{equation} \label{eq:linear_interpolation}
y = y_0 + \frac{x - x_0}{x_1 - x_0} (y_1 - y_0) = \frac{x_1 y_0 - x_0 y_1 + x y_1 - x y_0}{x_1 - x_0}.
\end{equation}

In this section, we construct a chemical reaction network (CRN) that computes the linear interpolation value through a sequence of oscillation-controlled reactions. We show that the value $y$ in \eqref{eq:linear_interpolation} can be computed using the modules described in Sections~\ref{sec:modules} and~\ref{sec:div-module}. To this end, we introduce the following notation:
\begin{equation} \label{def:notation_2}
\mathrm{Den} = x_1 - x_0, \qquad
\mathrm{Num} = x_1 y_0 - x_0 y_1 + x y_1 - x y_0.
\end{equation}
which represent the numerator and denominator in \eqref{eq:linear_interpolation}.

By a slight abuse of notation, throughout this section, we use the same symbols (e.g., $x_1^+$) to denote both the species and their corresponding concentrations.

\paragraph{Implementation.}
We describe a chemical reaction network (CRN) that computes the linear interpolation using dual-rail encoding.

\begin{enumerate}

\item \textbf{Dual-rail encoding.}
Each input is represented in dual-rail form. 

$(a)$ For the denominator $\mathrm{Den}$, we rewrite it as
\[
\mathrm{Den} = \underbrace{(x_1^+ + x_0^-)}_{:= \mathrm{Den}^+} \ - \ \underbrace{(x_1^- + x_0^+)}_{:= \mathrm{Den}^-}.
\]

$(b)$ For the numerator $\mathrm{Num}$, we rewrite it as
\begin{equation} \notag
\begin{split}
\mathrm{Num} & = \underbrace{(x_1^+ y_0^+ + x_1^- y_0^-) + (x^+ y_1^+ + x^- y_1^-) + (x_0^+ y_1^- + x_0^- y_1^+) + (x^+ y_0^- + x^- y_0^+)}_{:= \mathrm{Num}^+} 
\\& \quad - \underbrace{(x_1^+ y_0^- + x_1^- y_0^+) + (x^+ y_1^- + x^- y_1^+) + (x_0^+ y_1^+ + x_0^- y_1^-) + (x^+ y_0^+ + x^- y_0^-)}_{:= \mathrm{Num}^-}.
\end{split}
\end{equation}

\item \textbf{Computation of numerator and denominator.}

$(a)$ To compute the positive component $\mathrm{Num}^+$, we implement the combined multiplication and addition module (Section~\ref{sec:combined-mul-add}) during the $O_1$-active phase via the following network:
\begin{equation} \notag
\begin{split}
x_1^+ + y_0^+ + O_1 & \overset{1}\longrightarrow x_1^+ + y_0^+ + \mathrm{Num}^+ + O_1 \\
x_1^- + y_0^- + O_1 & \overset{1}\longrightarrow x_1^- + y_0^- + \mathrm{Num}^+ + O_1 \\
x^+ + y_1^+ + O_1 & \overset{1}\longrightarrow x^+ + y_1^+ + \mathrm{Num}^+ + O_1 \\
x^- + y_1^- + O_1 & \overset{1}\longrightarrow x^- + y_1^- + \mathrm{Num}^+ + O_1 \\
x_0^+ + y_1^- + O_1 & \overset{1}\longrightarrow x_0^+ + y_1^- + \mathrm{Num}^+ + O_1 \\
x_0^- + y_1^+ + O_1 & \overset{1}\longrightarrow x_0^- + y_1^+ + \mathrm{Num}^+ + O_1 \\
x^+ + y_0^- + O_1 & \overset{1}\longrightarrow x^+ + y_0^- + \mathrm{Num}^+ + O_1 \\
x^- + y_0^+ + O_1 & \overset{1}\longrightarrow x^- + y_0^+ + \mathrm{Num}^+ + O_1 \\
\mathrm{Num}^+ + O_1 & \overset{1}\longrightarrow O_1
\end{split}
\end{equation}
The network for the negative component $\mathrm{Num}^-$ is constructed analogously and in parallel.

\smallskip

$(b)$ Simultaneously, the positive component $\mathrm{Den}^+$ is computed using the addition module (Section~\ref{sec:addition}) during the $O_1$-active phase via the following network:
\begin{equation} \notag
\begin{split}
x_1^+ + O_1 & \overset{1}\longrightarrow x_1^+ + \mathrm{Den}^+ + O_1 \\
        x_0^- + O_1 & \overset{1}\longrightarrow x_0^- + \mathrm{Den}^+ + O_1 \\
        \mathrm{Den}^+ + O_1 & \overset{1}\longrightarrow O_1
\end{split}
\end{equation}
The network for the negative component $\mathrm{Den}^-$ is constructed analogously and in parallel.

\item \textbf{Computation of the interpolated value.}
We apply the generalized division module (Section~\ref{sec:generalized-div}) to obtain the interpolated value $y$, given by
\[
(y^+, y^-)
= \texttt{Generalized Division} \big( \mathrm{Num}^+,\mathrm{Num}^-,\mathrm{Den}^+,\mathrm{Den}^- \big).
\]
\end{enumerate}

\begin{example}

In this example, we compute the values of $y$ for different $x$ using two reference points $(x_1,y_1)=(1,2)$ and $(x_2,y_2)=(3,6)$, with $x$ ranging from $0.5$ to $3.5$. We compare these values with the exact interpolation results. Note that under mass-action kinetics ODEs, the modules converge asymptotically to their mathematical values, which introduces unavoidable approximation error~\cite{vasic2020crn++}.

\begin{figure}[H]
    \centering
    \includegraphics[width=0.5\linewidth]{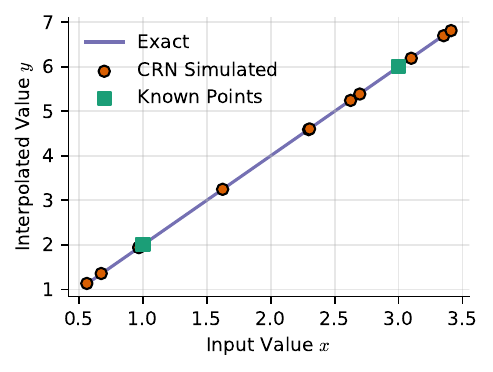}
    \caption{CRN-based linear interpolation: comparison between CRN-computed values and exact interpolation results.}
    \label{fig:interpolation_result}
\end{figure}
\end{example}

\section{Discussion} \label{sec:discussion}

In this work, we proposed a novel algorithm to compute division for signed numbers. Utilizing dual-rail encoding within chemical reaction networks, we established a theoretical framework capable of computing weights for linear regression and coefficients for linear interpolation. This enables the robust application of standard statistical inference techniques within mass-action dynamics. 

To validate our theoretical framework, we performed deterministic mass-action simulations on synthetic datasets. By implementing these simulations in Python, we provide an accessible codebase for prototyping future CRN algorithms that serves as an alternative to the Mathematica-based CRN++ framework. However, while these results confirm the mathematical correctness of the proposed reaction pathways, we acknowledge that physical realizations will inherently encounter stochastic fluctuations and thermodynamic constraints not captured by deterministic models.

Future work will focus on extending the theoretical foundation established here to more complex and non-linear machine learning paradigms. While the current framework successfully implements linear statistical inference, the dual-rail architecture and generalized division module could naturally be adapted for polynomial regression and higher-order non-linear function approximation. Furthermore, integrating these arithmetic modules into multi-layer chemical neural networks could enable the precise evaluation of complex activation functions within a continuous system. Finally, replacing the external readout of the Hopf oscillator with a fully intrinsic biochemical clocking mechanism remains a critical next step toward realizing completely autonomous molecular computing architectures.

%}
\section{Code Availability}
The source code for the plots and simulations for CRN implementations are publicly available on GitHub at \url{https://github.com/AryanK37/crn-linreg-interp}.

\bibliographystyle{plainnat}
\bibliography{sn-bibliography}

\end{document}